\documentclass[journal]{article}                                  

\usepackage{arxiv}
\usepackage{hyperref}
\usepackage{graphicx}          
\usepackage{amsmath} 
\usepackage{amssymb}  

\usepackage{amsthm}
\usepackage{color,soul}
\usepackage{flushend}
\usepackage{cite}
\usepackage{mfirstuc}
\usepackage{mathtools}
\usepackage{xcolor}
\usepackage{array,multirow}
\usepackage{algorithm} 
\usepackage{algorithmic}  
\usepackage[linesnumbered,algo2e,ruled,vlined,norelsize]{algorithm2e} 

\allowdisplaybreaks
\newtheorem{thmm}{Theorem}
\newtheorem{propp}{Proposition}

\newtheorem{assm}{Assumption}

\usepackage{tikz}
\usepackage{textcomp}
\usepackage{hyperref}
\usepackage{lipsum}
\newtheorem{definition}{Definition}

\usepackage{authblk}
\title{\LARGE \bf
Improving Fast Charging Safety With Core Temperature Estimation Via Kolmogorov-Arnold Network
}

\author[1]{Faysal Ahamed}
\author[1]{Tanushree Roy}
\affil[1]{Department of  Mechanical and Aerospace Engineering, Texas Tech University, Lubbock, TX 79409, US. Emails:~{\tt\small fahamed@ttu.edu, tanushree.roy@ttu.edu}.}

\begin{document}


\maketitle
\pagestyle{empty}

\begin{abstract}
Fast charging of Lithium-ion batteries can lead to a significant temperature rise, which can cause serious risks to battery safety and lifetime. To ensure safe battery operation, thermal constraints must be enforced during the fast charging process. However, the core temperature of the battery cannot be directly measured in practice, which makes real-time safety enforcement challenging. This paper proposes a framework that incorporates core temperature estimates from Kolmogorov-Arnold Network within robust control barrier function (KAN-rCBF) constraints for battery fast-charging. The algorithm utilizes measurements from battery surface temperature, coolant temperature, coolant power, and charging current to solve a quadratic programming problem under safety constraints. We prescribe analytical safety guarantees for this optimal charging policy under KAN estimation errors and model uncertainty. Simulation results show that the proposed method maintains a safe battery temperature while achieving charging times comparable to the state-of-the-art method, where the latter fails to guarantee the same level of thermal safety.
\end{abstract}

\section{Introduction}

Lithium-ion batteries are central to our energy future, as they can be recharged and can store a significant amount of energy efficiently. However, in many applications, including electric vehicles, the primary challenge has shifted from energy storage capacity to charging performance. Specifically, the focus is on achieving rapid and safe charging while preventing thermal runaway and material degradation. To tackle this problem, researchers have introduced different strategies, which can be broadly divided into three groups: heuristic, data-driven, and model-based methods.

\subsection{Literature}

Recent studies have explored heuristic approaches like multi-stage constant current (MSCC), constant power (CP), and pulse charging (PC) to improve the speed of charging and safety of lithium-ion batteries. For instance, Taguchi-optimized MSCC protocols have been explored in multiple studies for battery fast charging \cite{vo2015new,jiang2020optimization,tahir2024multi}. Similarly, \cite{han2024optimal} proposed a Dandelion optimizer to design an optimal MSCC approach that also reduced charge time without overheating the batteries. An expansion force-based adaptive MSCC scheme was introduced in \cite{shen2025expansion}, which achieved nearly 50\% faster charging without capacity loss. In addition,  \cite{liu2024experimental} tested three pulse charging modes: charge-charge (C-C), charge-rest (C-R), and charge-discharge (C-D) on Lithium iron phosphate cathode (LFP) and Nickel–manganese–cobalt oxide cathode (NMC) cells to show that the C-D mode reduces charging time and lowers resistance, and extends cycle life relative to CC-CV. \cite{guo2024unravelling} showed that pulse current charging significantly extends the cycle life of graphite and NMC532 electrodes by promoting uniform Li-ion distribution and by improving stabilization of electrode interfaces. Additionally, \cite{tahir2023overview} showed that improvement in efficiency, battery life, and thermal regulation in MSCC are heavily dependent on the design of charging stages. Altogether, heuristic strategies improve speed and safety compared to CC-CV, but they are limited, as the fixed rules fail to adapt to changing conditions.

In recent years, data-driven methods have also been explored to improve fast charging capabilities beyond heuristic rules. For instance, deep reinforcement learning with electrothermal modeling has been used in \cite{zhang2025fast} to decrease the charging time compared to CC-CV and ensure safe voltage and temperature limits. Similarly, \cite{shokry2025health} designed a health-constrained predictive charging controller using deep learning networks for fast charging while maintaining the desired temperature. In addition, \cite{sayed2025optimizing} proposed an adaptive reinforcement learning-based fast charging framework that reduces charging time under voltage and thermal constraints, and showed that this method can balance fast charging with improved health and safety. Finally, \cite{yu2025optimized} used a deep reinforcement learning framework to optimize MSCC charging with safety constraints. 

Model-based control uses the physics-based dynamics of the batteries to ensure safer and more reliable fast charging. \cite{tian2020real} implemented a nonlinear double-capacitor battery model in an explicit model predictive control (MPC) framework for lithium-ion cells to ensure fast, health-aware charging with low computational cost. Similarly, \cite{bills2022model} proposed a quadratic battery model and validated its performance against high-fidelity models in an MPC framework for safe fast charging under thermal, voltage, and safety constraints. Alternatively, \cite{li2021electrochemical} combined a computationally-efficient ensemble transform Kalman filter with a proportional-integral controller to ensure safety and efficiency in fast charging. In addition, \cite{hirt2025safe} combined Bayesian optimization and model predictive control to achieve faster charging than the standard model predictive control with voltage and temperature limits. Moreover, analytical design criteria for battery charging controllers have been derived using Lyapunov and barrier function theory to keep the battery thermal systems stable and within safe operating limits using lumped thermal models \cite{vyas2022thermal,vyas2024input} and heat equations \cite{roy2024input}. \cite{feng2024safe} proposes a control barrier function (CBF) framework for safe battery charging and discharging under SoC, terminal voltage, and surface temperature constraints by splitting the CBF problem into cascaded electrical and thermal subsystems, which results in sequential Quadratic Program (QP) problems that improve computational efficiency while maintaining safety constraints.

\subsection{Research gap} 
Despite the recent advances, several challenges remain in guaranteeing the safety of battery fast-charging control. 
Heuristic fast-charging strategies determine charging profiles using predefined rules or offline optimization and rely on fixed charging patterns. Therefore, these methods have limited ability to adapt to internal battery state variations during real-time operation. 
On the other hand, data-driven approaches can learn accurate charging protocols more effectively from operational battery data and enable faster policy evaluation. However, purely data-driven methods often lack interpretability and analytical guarantees of performance under system uncertainty. 
In contrast, model-based fast-charging policies leverage physics-based battery models, which enhance the interpretability of the resulting control decisions and often incorporate safety constraints that can be enforced with analytical guarantees. Now, to improve the efficacy of the model-based control, battery state estimation (such as core temperature) is essential for unmeasured states. But model-based estimations are considerably slower than machine-learning-based predictions, which can be obtained via a forward pass of the network \cite{ma2020core}. This model-based estimation bottleneck is especially crucial when the estimates are used within an optimization framework, such as MPC or CBF, which are themselves computationally expensive. 

Therefore, a key research gap remains in developing an interpretable fast-charging framework that leverages the computational efficiency and accuracy of machine learning based estimation and the robust analytical performance guarantees of model-based control theory.

\subsection{Need for Kolmogorov-Arnold network (KAN)}
Accurate estimation of battery temperature requires models that can capture the nonlinear relationship between operating conditions and thermal behavior. Machine learning methods are often used for this purpose because they can learn these relationships directly from data.  Among these methods, traditional neural networks such as multi-layer perceptrons (MLP, LSTM, RNN) are widely used \cite{goodfellow2016deep}. However, these models are often extremely large and are used as black-box systems, as it is difficult to understand how each input variable affects the output through the large model structure. Model-based thermal estimation methods rely on accurate battery models and parameter identification. However, variations in thermal parameters and operating conditions can affect estimation accuracy and make real-time implementation more difficult \cite{dong2021analysis}.

To address these limitations, this work employs the Kolmogorov-Arnold network (KAN).
Unlike conventional neural networks that learn linear combinations of weights with fixed nonlinear activation, KAN models learn nonlinear activation functions along the edges of the network to be summed up at the nodes \cite{liu2024kan}.
This structure provides two main advantages. First, KAN can represent complex nonlinear relationships with fewer parameters, which results in a computationally efficient and compact model structure. Second, because the learned functions admit symbolic representations, the compact model allows a tractable and interpretable functional mapping from the battery input variables to the estimated core temperature. \cite{mallick2025kan} showed that KAN-based thermal models provide accurate temperature estimation with smaller model structures and faster estimation time relative to the best-performing neural networks and faster estimation compared to model-based implementation. 
Moreover, KAN has been used for estimation of state-of-charge (SoC)  \cite{sulaiman2024battery}, state-of-health (SoH) \cite{cui2025enhanced,shao2025soh,liu2025soh,ghosh2026explainable}, capacity degradation\cite{yang2025capacity}, diagnostics of battery thermal anomalies \cite{ghosh2026kankoopman}, and remaining useful life (RUL) \cite{he2025remaining} of lithium-ion batteries.
This highlights KAN's potential use for safety-critical applications, such as battery fast charging, as well as its implementability on physical battery management systems with limited computational resources. In addition, the KAN approximation theorem provides a theoretical bound for KAN estimation error, which has been leveraged to derive the analytical safety guarantees for our proposed safety framework. Therefore, the combination of interpretability, computational efficiency, and provable approximation capability of KAN allows accurate estimation of the battery core temperature in this work.

Our primary contributions of this paper are as follows.
\begin{enumerate}
    \item We integrated a KAN-based core temperature estimator within the control barrier function framework to develop a fast charging strategy that maintains safe battery temperature. 
    \item We obtain a charging policy that optimally satisfies the barrier function constraints on battery temperatures, voltage, and state-of-charge (SoC) of the battery while remaining close to the high C-rate reference.
    \item We derive the analytical conditions that guarantee the performance of our fast charging policy under model uncertainty and KAN estimation error.
\end{enumerate}
The rest of the paper is organized as follows. Section~\ref{prel} presents the preliminaries, i.e., KAN, lumped parameter thermal model, and estimating core temperature using KAN. In Section~\ref{fast charging}, control barrier function for fast charging, model development, and KAN-rCBF charging optimization conditions are described. Next, we present our simulation results in Section~\ref{sim}. Finally, in Section~\ref{conclu} we conclude our work. 

\noindent
\textbf{Notation:} \label{TH}
The $C^m$-norm of a function $g(x)$ on the interval $I$ is given by $\|g\|_{C^m} = \max_{|\Psi|\leqslant m} \sup_{x \in I}  \bigl| D^\Psi g(x) \bigr|$, where $D^\Psi$ is the partial derivative operator for  multi-index $\Psi$. For a scalar quantity $a$, $|a|$ denotes its absolute value. Given a function $h(x)$ and a vector $v(x)$, the Lie derivative of $h$ along $v$ is defined as $ L_v h(x) := \nabla h(x)^\top v(x)$. A continuous function $\alpha:[0,a)\rightarrow[0,\infty)$ is called a class-$\mathcal{K}$ function when it is strictly increasing and satisfies $\alpha(0)=0$. A function $\alpha:(-b,a)\rightarrow \mathbb{R}$, where $a,b>0$, is called an extended class-$\mathcal{K}$ function if it is continuous, strictly increasing, and satisfies $\alpha(0)=0$.

\section{Preliminaries} \label{prel}
This section introduces the key components used in this work. First, KAN is briefly described and the Kolmogorov-Arnold approximation theorem is introduced. Next, the electro-thermal battery model and the core temperature estimation approach based on the KAN model are presented.
\subsection{KAN structure} \label{KAN}

The Kolmogorov–Arnold representation theorem states that for any multivariate continuous function on a bounded domain, there exists a finite collection of continuous univariate functions, which through additive composition can represent the former \cite{liu2024kan}.
By leveraging learning techniques, the KAN can approximate univariate continuous functions to represent nonlinear multivariate functional relationships, thereby harnessing the theorem’s representational power. Each KAN layer has input nodes and output nodes. Each input node passes through a learnable (univariate) activation function, which is then added to produce the value of an output node. The output nodes of one KAN layer become the input for the next KAN layer, forming a multi-layer model.

Mathematically, if a KAN has $M$ layers, then each layer is denoted by $n$ for $n \in \{1\dots M\}$. Subsequently, each layer $n$ can have $\beta_n$ number of input nodes and $\beta$ number of output nodes. The outputs of layer $n$ become the inputs of the next layer $(n+1)$. The value of the $p$-th input node in layer $n$ is written as $z_{n,p}$,  $\forall p \in \{1\dots \beta_n\}$, and the value of the $c$-th output node in the next layer is written as $z_{n+1,c}$,  $\forall c \in \{1\dots \beta_{n+1}\}$. The inputs $z_{n,p}$ pass through activation functions $\theta_{n,c,p}$, are then added to obtain the value of output nodes $z_{n+1,c}$. In compact matrix notation, we can write

\begin{equation} \label{KAN matrix}
   \text{z}_{n+1} = \underbrace{ \begin{pmatrix}
\theta_{n,1,1}(\cdot) & \dots & \theta_{n,1,\beta_n}(\cdot)\\
\vdots & \ddots & \vdots\\
\theta_{n,\beta_{n+1},1}(\cdot) & \dots & \theta_{n,\beta_{n+1},\beta_n}(\cdot)
    \end{pmatrix}}_{\Theta_{n}} \text{z}_n,
\end{equation}
\begin{align}
 \text{ where }    \text{z}_{n+1}=  \begin{pmatrix}
       {z}_{n+1,1}\\
       \vdots\\
       {z}_{n+1,\beta_{n+1}}
   \end{pmatrix} \, \text{ and } \text{z}_n=\begin{pmatrix}
        z_{n,1}\\
        \vdots\\
        z_{n,\beta_n},
    \end{pmatrix}
\end{align}
\noindent and $\Theta_n$ is the matrix activation function that learns the intrinsic non-linear structure in data and captures the relationship between the input variables and the output variable. Specifically, an $M$-deep KAN model with  $j$ input variables denoted by $\text{z}_i = [z_{i,1}\dots z_{i,j}]$ and a single output variable denoted by $\text{z}_o$ represent the functional relationship between $\text{z}_i$ and $\text{z}_o$ represented by, $\text{z}_o = (\Theta_{M} \circ \Theta_{M-1} \circ \cdots \circ \Theta_{1})\text{z}_i$. These activation functions $\Theta_{n}$ at each layer $n$ are represented as the composition of basis functions and learnable B-spline functions of $k-$th order over  $G$ grid points. The loss function for our KAN model consists of prediction errors, activation function entropy, and their $l_1$ regularization. 

Furthermore, we will use the following Kolmogorov–Arnold approximation theorem to derive analytical conditions to ensure the performance of the integrated KAN-rCBF.

\begin{thmm}[Kolmogorov–Arnold Approximation Theorem]\label{KAT}
  Let a multi-variate function $h(z)$ for $z=(z_1,...,z_j)$  be written as a composition of $k+1$ times differentiable univariate activation functions $\Theta_j, \forall j\in \{1,\hdots, M\}$ i.e. $h(z) = \bigl(\Theta_{M} \circ \Theta_{M-1} \circ \cdots \circ \Theta_{1}\bigr)(z)$. Let an $M$-layer KAN-network is learned with $G$ grid points and $k$-th order B-splines to represent this function and is denoted by $\bigl(\Theta_{M}^G \circ \Theta_{M-1}^G \circ \cdots \circ \Theta_{1}^G\bigr)(z).$ The error in approximation for the KAN-network is then given by:
  \begin{align}\label{KAN-approx}
      \bigl\| h - \bigl(\Theta_{M}^G \circ \Theta_{M-1}^G \circ \cdots \circ \Theta_{1}^G\bigr) \bigr\|_{C^m}
\leqslant \mathcal{M} G^{-(k+1-m)}, 
  \end{align}
 $\forall m \in \{1, \hdots, k\}$ and the constant $\mathcal{M}$ depends on the function $h$ and its smoothness.
\end{thmm}

Essentially, Theorem~\ref{KAT} expresses the error bound in estimating a function $h$ using the multi-layer KAN in terms of its structure.

\subsection{Estimating core temperature using KAN}
The electro-thermal battery dynamics \cite{vyas2022thermal} can be written in terms of battery states  $x= [\text{SoC}, T_c, T_s, T_{\infty}]^T$ as  
\begin{align}
\dot{x} = f(x) + g(x)U+ \nu(t) , \quad V  = V_o(\text{SoC})-uR_0,  \label{modelg} 
\end{align}

\noindent
where the battery dynamics $f(x)=\big[ 0, \,
-\tfrac{T_c - T_s}{R_c C_c}, \,
-\tfrac{T_s-T_c}{R_cC_s} - \tfrac{T_s-T_{\infty}}{R_sC_s}, \,
-\tfrac{T_{\infty}-T_s}{R_sC_{\infty}}\big]^T$; function $g(x) = [g_1(x), g_2, g_3]^T$ with $
g_1(x)=
\big[
-\tfrac{1}{C_{\text{bat}}}, \,
- \tfrac{\mathcal{E} T_c}{C_c}, \,
0, \,
0
\big]^T$, $g_2=
\big[
0, \,
\tfrac{R_0}{C_c} , \,
0, \,
0
\big]^T
$,  and  $g_3 =\left[0, \, 0, \, 0, \, -\tfrac{1}{C_{\infty}}\right]^T$. Here, $T_c, T_s,\, \text{and}\,\, T_{\infty}$ are the core, surface, and environment temperatures, respectively. The input vector $U = [u, u^2, \dot{Q}_c]^T$, where $u$ is the control input current, which is negative while charging and positive while discharging, and $\dot{Q}_c$ is the cooling power. The model uncertainties are represented by $\nu(t)= [\nu_{\text{SoC}}(t), \nu_c(t), \nu_s(t), \nu_{\infty}(t)]^T$, where the elements represent uncertainty in the SoC, core, surface, and ambient temperature dynamics. $C_{bat}$ is the battery capacity, $R_0$ is the ohmic resistance, and $R_c$, and $C_c$ are the thermal resistance and capacitance of the cell core, respectively. $R_s, C_s$ describe the resistance and capacitance related to heat transfer at the surface. $C_{\infty}$ is the heat capacity of the cooling system.  $\mathcal{E}$ is the entropic heat coefficient. The output equation is for the battery terminal voltage $V $, which uses the SoC-dependent open-circuit voltage $V_{o}(\text{SoC})$.  

From \eqref{modelg}, we can observe that there exists a function $h$ such that ${T}_c = h(I,\dot{Q}_c,T_s,T_{\infty})$. Moreover, $I$, $\dot{Q_c}$, $T_s$, and $T_{\infty}$ are readily measurable and are thus used as input features for our KAN-based core temperature model to estimate the core temperature $\widehat{T}_c$. Mathematically, the KAN core temperature estimator learns the following map:  
\begin{align} \label{KAN network}
  \widehat{T}_c = (\Theta_{M-1}^G \circ \cdots \circ \Theta_{1}^G\bigr)(I,\dot{Q}_c,T_s,T_{\infty}).  
\end{align}
To train this KAN-based thermal model, we considered a $2.3Ah$ $A123$ cylindrical $LiFePO_4-LiC_6$ battery cell parameters \cite{mallick2025kan} to generate the data from the lumped thermal model \eqref{modelg} using both constant and dynamic current profiles. A small Gaussian noise term $5\times10^{-4}\xi$, where $\xi \sim \mathcal{N}(0,1)$, is added to the thermal dynamics to represent model uncertainty.

The specific hyperparameter of our trained KAN structure is shown in Table~\ref{tab:MODEL_VARIABLES_AND_PARAMETERS}.

\begin{table}[ht]
    \centering
    \renewcommand{\arraystretch}{1.4} 
    \setlength{\tabcolsep}{8pt}
    \begin{tabular}{|c|c||c|c|}
    \hline
    Network width & $[4,1,1]$ & Training Epochs & $70$   \\
    \hline
     $G$ & $2$ & Batch  size & $ 2^{14}$ \\
    \hline
    $k$ & $3$  & Grid update stop & $50$  \\
    \hline
    $\text{Regularization} $ & $0.00005$ & Optimizer & L-BFGS  \\
    \hline
    $ l_1\text{ penalty}$ & $0.14$ & Entropy penalty & $0.17$   \\
    \hline
    \end{tabular}
    \vspace{1mm}
    \caption{Model structure and hyperparameters for the KAN used for core temperature estimation.}
    \label{tab:MODEL_VARIABLES_AND_PARAMETERS}
\end{table}

The KAN estimation error bound \eqref{KAN-approx} and the electro-thermal model \eqref{modelg} discussed in this section will be subsequently used to derive the robust control barrier function constraints in the next section.

\section{Integration of KAN and Robust CBF} \label{fast charging}

In this section, we will describe how the KAN-rCBF algorithm can be implemented to generate safe fast charging protocols using the KAN-estimated core temperature $\widehat{T}_c$ and the measurements data of surface temperature $T_s$, coolant temperature $T_{\infty}$, current $I$, and coolant power $\dot{Q}_c$. The scheme overview in Fig.~\ref{fig:overview} shows how the measured battery data is used to obtain the core temperature estimate from KAN and is subsequently sent to the KAN-rCBF policy optimizer along with the KAN estimate to obtain the optimal safe charging policy.

\begin{figure}[ht]
    \centering
    \includegraphics[width=.7\linewidth]{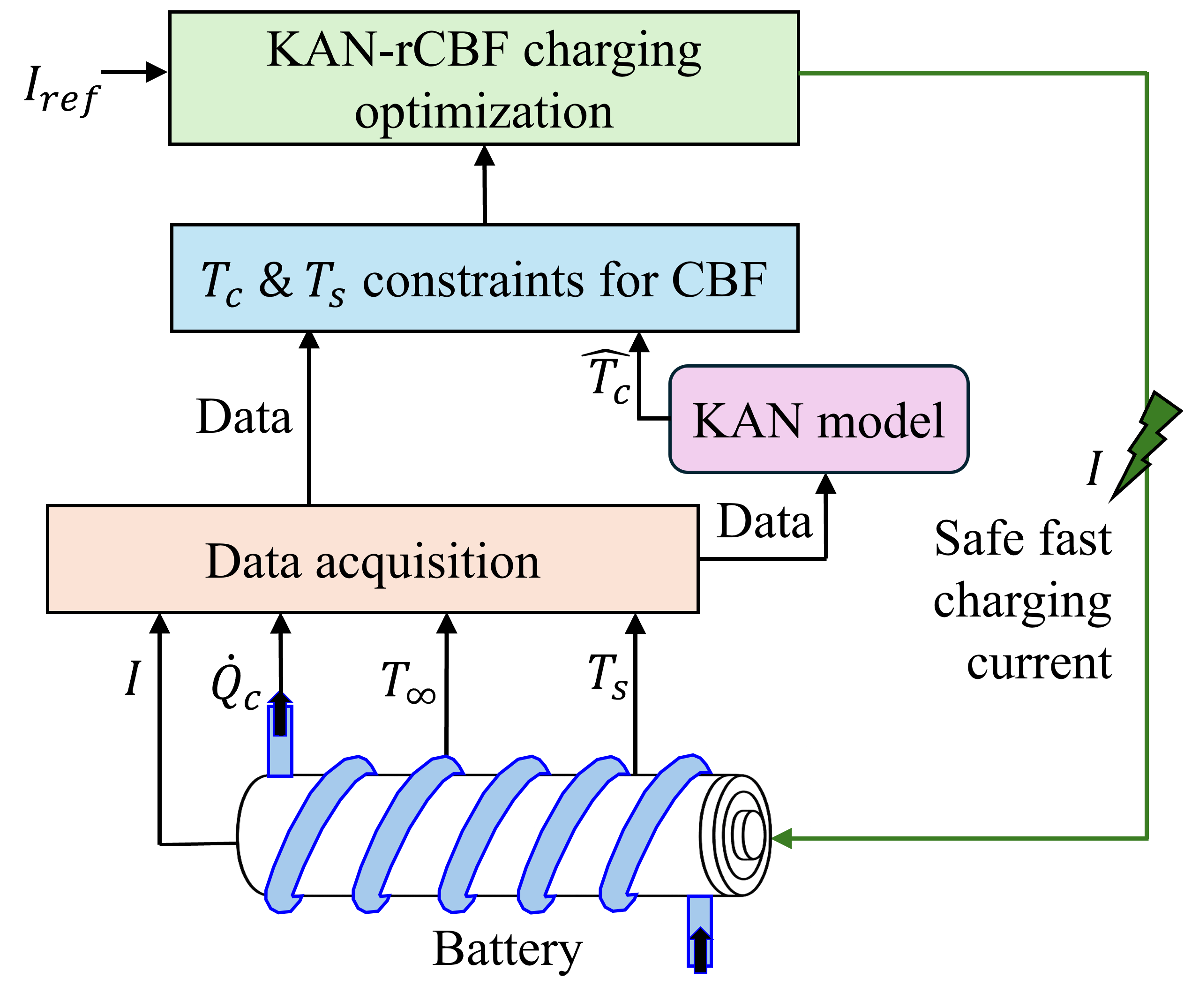} 
    \caption{Overview of the CBF safety-constrained fast charging with KAN core temperature estimation $\widehat{T}_c$ using $T_s, T_\infty, I$ and $\dot{Q}_c$.}
    \label{fig:overview}
\end{figure}

Before we present the KAN-rCBF charging optimization strategy, we will first state the assumptions in our framework and also define the relevant mathematical preliminaries.

\begin{assm}[Bounded system uncertainties] \label{Assumption2}
    The uncertainties $\nu_{\text{SoC}}(t),\,\nu_c(t),\, \nu_s(t), \,  \nu_{\infty}(t) \, \text{and} \, \dot{\nu}_s(t)$ in the battery model \eqref{modelg} are bounded for all $t \geqslant 0$. That is, there exist known positive constants $\overline{\nu}_{\text{SoC}},\,\overline{\nu}_c,\, \overline{\nu}_c,\, \overline{\nu}_{\infty}\, \text{and} \, \overline{d\nu_s}$ 
    such that $\forall t$,
    \begin{align*}
        |\nu_{\text{SoC}}(t)| &\leqslant \overline{\nu}_{\text{SoC}},\,\,
        |\nu_c(t)| \leqslant \overline{\nu}_c, \,\, |\nu_s(t)| \leqslant \overline{\nu}_s, \,\, \\ &|\nu_{\infty}(t)| \leqslant \overline{\nu}_{\infty} , \,\, |\dot{\nu}_s(t)| \leqslant \overline{d\nu_s}. 
    \end{align*}
\end{assm}

\begin{definition}[Safe set]
    For a system with states $x$ and input $U$, a safe set for the system can be defined in terms of a barrier function $B(x)$ that is continuously differentiable in $x$ as, 
    \begin{align}
        \mathcal{S} = \{x : B(x) \geqslant 0\}, \label{safetyset}
    \end{align}
    where $B(x) > 0$ represents safe operating conditions, while $B(x) = 0$ denotes the safety boundary.
\end{definition}

\begin{definition}[Forward invariance]
    A set $\mathcal{S}$ is a forward invariant for the system \eqref{modelg} if a starting solution $x(0) \in \mathcal{S}$ remains in $\mathcal{S}$ for all time $t\geqslant 0$.
\end{definition}
A battery is considered to charge safely when its state of charge (SoC), terminal voltage $V $, core temperature $T_c$, and surface temperature $T_s$ does not exceed the maximum allowable state of charge $\overline{\text{SoC}}$, terminal voltage $\overline{V}$, core temperature $\overline{T}_c$, and surface temperature $\overline{T}_s$, respectively.  Using these limits we define continuously differentiable barrier functions as
\begin{align}
    &B_1(x,u)=\overline{V}-V  = \overline{V} -V_o(\text{SoC})+uR_o,   \label{V11}\\
    &B_2(x)=\overline{\text{SoC}}-\text{SoC}, \label{soc s}\\
    &B_3(x)=\overline{T}_c-T_c, \label{llls}\\
    &B_4(x)= \overline{T}_s-T_s, \label{Ts}
\end{align}
to obtain the safe set for the fast charging control as
\begin{align}
    \mathcal{S}=\{x:B_1(x,u)\geqslant 0, B_2(x)\geqslant 0, B_3(x) \geqslant 0,B_4(x) \geqslant 0\}. \label{Safety set}
\end{align} 
The next theorem prescribes the conditions on the barrier functions $B_i(x), \, \forall i\in\{2,3,4\}$  defined in \eqref{soc s}-\eqref{Ts} such that the safe set \eqref{Safety set} is forward invariant. Alternatively, it prescribes the condition that the control input $u$ must satisfy to ensure any battery state that starts in the safe set will remain within the safe set for all time $t\geqslant 0$ \cite{ames2016control}. We note here that $B_1(x,u)$ depends explicitly on the input (i.e. it has relative degree $r=0$). Therefore, any input satisfying $B_1(x,u)\geqslant 0$ is sufficient to guarantee forward invariance of $S$, with no further conditions required.

\begin{thmm}[Constraints for rCBF \label{theory2}]
Consider a system $\dot{x}=f(x)+g(x)U+\nu(t)$ as in \eqref{modelg} and some safe set defined in \eqref{safetyset}, then a function $B(x)$ will be a barrier function for $r\geqslant 1$ for the system if $\forall r\geqslant 1$
\begin{align}
L_f^rB(x)+L_gL_f^{r-1}B(x)U+\Delta_r(x,t)\geqslant -K_{\alpha}\eta(x). \label{ecbf}
\end{align}
where the uncertainty term is defined as $\Delta_r(x,t)=\sum_{k=0}^{r-1}L_{\nu^{(k)}}L_f^{r-1-k}B(x)$,  the barrier state vector $\eta(x) = [B(x),\, L_f B(x),\, \dots,\, L_f^{r-1} B(x)]^{T}$ and the elements in $K_{\alpha}\in \mathbb{R}^r$ are such that main minors of the determinant of the corresponding Hurwitz matrix constructed from $K_\alpha$ are positive. The proof is provided in the Appendix~\ref{proof_hocbf}. If such a barrier function exists, then the safety of the set \eqref{safetyset} is forward invariant. 
\end{thmm}

To obtain a control policy $u$ that matches a desired policy $u_{des}$ under the rCBF conditions \eqref{ecbf}, we will solve the following optimization problem  
\begin{align}
U^* = \arg\min_{U} \|U - U_{\text{des}}\|^2,\, \forall t\geqslant 0. 
\end{align}
For fast charging, we want the control policy to be as close to the highest C-rate current $I_\text{ref}$ as possible. In terms of the battery dynamics, we can rewrite the rCBF constraints \eqref{ecbf} in terms of the charging current $u$ as
\begin{align}
&L_f^r B(x) +L_{g_1}L_f^{r-1} B(x)u+L_{g_2} L_f^{r-1} B(x)u^2 +L_{g_3}L_f^{r-1} B(x)\dot{Q}_c  +\Delta_r (x,t) \geqslant -K_{\alpha}\eta(x).  \label{robust_co}
\end{align}
For applying rCBF constraints \eqref{robust_co} to our battery system, we first determine the relative degree of each barrier function $B_i$ with respect to the control input $u$. Using the electro-thermal model \eqref{modelg}, we can readily derive that the voltage, SoC, core temperature, and surface temperature barrier functions have relative degrees 0, 1, 1, and 2, respectively. 

We will now present the analytical safety guarantees of our KAN-rCBF method in Proposition~1. 

\begin{table}[t]
\centering
\renewcommand{\arraystretch}{2.3} 
\begin{tabular}{|*{3}{c|}}
\hline

$A_1=\tfrac{1}{R_cC_c}$ & 
$A_2=\tfrac{\mathcal{E}}{C_c}$ &
{$A_3=\tfrac{R_0}{C_c}$}

\\ \hline

\multicolumn{3}{|l|}{
$ A_4=
\tfrac{1}{R_c^2C_cC_s}
+\tfrac{1}{R_c^2C_s^2}
+\tfrac{1}{R_cR_sC_s^2}
-\tfrac{\alpha_4}{R_cC_s}
$}
\\ \hline

\multicolumn{3}{|l|}
{$A_5= -A_4 -\tfrac{1}{R_cR_sC_s^2} -\tfrac{1}{R_s^2C_s^2} -\tfrac{1}{R_s^2C_sC_{\mathrm{bat}}}+\tfrac{\alpha_4}{R_sC_s}
-\alpha_3
$}  \\ \hline

\multicolumn{2}{|l|}
{$
A_6=
\tfrac{1}{R_cR_sC_s^2}
+\tfrac{1}{R_s^2C_s^2}
+\tfrac{1}{R_s^2C_sC_{\mathrm{bat}}}
-\tfrac{\alpha_4}{R_sC_s}
$} &$A_7=\alpha_3\overline{T}_s$
\\ \hline
\multicolumn{3}{|l|}
{$

A_8=\tfrac{\mathcal{E}}{R_cC_cC_s}\, A_9=-\tfrac{R_s}{R_cC_cC_s}
\,
A_{10}=\tfrac{1}{R_sC_sC_{\mathrm{bat}}}
$}
\\ \hline

\multicolumn{3}{|l|}
{$
{\Delta}_4\!=\!{}
-\!\left(\tfrac{1}{R_cC_s}\right)\overline{\nu}_c
\!-\!\left|\tfrac{1}{R_cC_s}+\tfrac{1}{R_sC_s}\!-\!\alpha_4
\right|\overline{\nu}_s\!-\!\left(\tfrac{1}{R_sC_s}\right)\overline{\nu}_{\infty}
\!-\!\overline{d\nu_s}
$}
\\ \hline

\end{tabular}
\caption{Constants for Proposition~\ref{main_prop1}.} \label{tab:constants}
\label{coeff}
\end{table}
\begin{propp}[KAN-rCBF charging optimization condition]\label{main_prop1}
    Consider the electro-thermal battery system in \eqref{modelg}, whose states are the state of charge SoC, core temperature $T_c$, surface temperature $T_s$, and $\widehat{T}_c$ is the estimated core temperature using the KAN-based thermal model \eqref{KAN network}. Let the safe set $\mathcal{S}$ be defined as in \eqref{Safety set}. If the charging control input $u$ is computed as the solution of the following quadratic optimization problem
\begin{align} \label{optimization}
    &u_{safe} = \arg \min_{u\in \mathcal U}||u - I_\text{ref}||^2,\\
     s.t. \quad &\overline V  - V_{o}(\text{SoC}) + R_0u \geqslant 0,    \label{opt_constrain}\\
    &\tfrac{u}{C_{\text{bat}}} + \alpha_1(\overline{\text{SoC}}-\text{SoC})-\overline{\nu}_{\text{SoC}}\geqslant 0, \label{u_cond} \\
    &H_c(\widehat{T}_c,T_s,u)+\overline{\nu}_c\geqslant - \tfrac{\mathcal{E}}{C_c} u \mathcal{M}G^{-k-1} \label{Hc_cond} \\
    &H_s(T_c,T_s, T_{\infty},\dot{Q}_c,u)+ {\Delta}_4\geqslant 0 \label{SSC}\\
    & \mathcal{U}=\{u \in \mathbb{R}:I_{ref} \leqslant u \leqslant 0\}, \label{curr}
\end{align} 
where $H_c(\widehat{T}_c,T_s,u)=-A_1 T_s+ A_2 \widehat{T}_cu-A_3u^2 + A_1 \overline{T}_c$ and $H_s(T_c,T_s, T_{\infty},\dot{Q}_c,u)
   =  A_4 T_c + A_5 T_s + A_6 T_{\infty} + A_7 + A_8 T_c u + A_9 u^2 + A_{10} \dot Q_c$, the constants $A_m,\, \forall m\in\{1,\dots,10\}$, and ${\Delta}_4$ are given in Table~\ref{coeff}, then states of the battery will remain in safe set $\mathcal{S}$ for all time $t\geqslant 0$.  
\end{propp}

\begin{proof} 
The proof is presented in Appendix~\ref{proof}.
\end{proof}
\begin{algorithm}[ht]
\caption{KAN-rCBF Fast Charging Current Optimization Algorithm}\label{algorithm}

\KwIn{Time $t$, Initial current $u(0)$, coolant power $\dot{Q_c}$, coolant temperature $T_{\infty}(t)$, surface temperature $T_s(t)$, \texttt{KAN} model \eqref{KAN network} , $A_m,\, \forall m\in\{1,\dots,10\}$}
\KwOut{Optimized current $u$.}

\SetKwFunction{KT}{KAN\_Therm}
\SetKwFunction{CO}{rCBF}

\For{$t \geqslant 0$}{ 
    Forward-pass through \texttt{KAN} \eqref{KAN network} to get $\widehat{T}_c$\;
    $u(t+1) \gets$ \CO{$u(t),\dot{Q_c}, T_{\infty}(t), T_s(t), \widehat{T}_c(t)$}\;
}

\SetKwProg{Fn}{Function}{:}{}
\Fn{\CO{$ u(t), \dot{Q_c}, T_{\infty}(t), T_s(t)$}}{
    Compute $u_{safe}$ using \eqref{optimization} -- \eqref{curr}\;
    \textbf{return} $u_{safe}$\;
}

\end{algorithm}

\section{Simulation results} \label{sim}
This section presents simulation results that evaluate the proposed KAN-rCBF framework and compare it with the recent state-of-the-art surface temperature constrained CBF (S-CBF) method \cite{feng2024safe}. To evaluate the effectiveness of the proposed KAN-rCBF strategy and baseline, optimizations are solved for a desired reference $I_{ref}=23 A$ (i.e. 10C) fast-charging scenario. The implementation of Algorithm~\ref{algorithm} is based on the safety constraints developed in Proposition~\ref{main_prop1}. These safety constraints are enforced throughout the charging process, with $0 \leqslant V_t \leqslant \overline{V}_t = 3.6V$, $T_c,\,T_s \leqslant \overline{T} = 323K$, and $0 \leqslant \mathrm{SoC} \leqslant 1$.  The CBF parameters are chosen as $\alpha_1 = 0.05$, $\alpha_2 = 0.0246$, $\alpha_3 = 0.05$, and $\alpha_4 = 0.6$, providing a reasonable balance between responsiveness and robustness. In the simulations, the uncertainty bounds are selected as $\overline{\nu}_c = 0.0005$, $\overline{\nu}_s = 0.0005$, $\overline{\nu}_{\infty} = 0.0005$, and $\overline{d\nu_s} = 0.1$. These uncertainty bounds are incorporated into the CBF constraints through worst-case margins, ensuring that the safety requirements are satisfied during operation.
All computations have been conducted on a Dell Precision 3660 workstation equipped with an Intel Core i7-13700 CPU (16 cores, 24 threads), 32 GB RAM, and an NVIDIA RTX A2000 GPU with 12 GB VRAM. GPU acceleration was enabled using NVIDIA driver 573.44 with CUDA 12.8, and the implementation was developed in Python 3.13.5 using the PyTorch ecosystem. The optimal charging current was computed using the Sequential Least Squares Programming (SLSQP) algorithm from the SciPy optimization library. 
    In this work, SoC is generated from the battery model through current counting and thus the uncertainty $\nu_{\text{SoC}}=0$. However, if the SoC is estimated, this uncertainty can be incorporated using a similar strategy as used for the rest of the states. 

Before presenting the results for KAN-rCBF, we will present the accuracy of the KAN core temperature estimator.

\textbf{Performance of KAN:} The estimated temperature $\widehat{T}_c$ obtained from KAN and the true temperature $T_c$ are shown in Fig.~\ref{error}. The lower plot in Fig.~\ref{error} shows the error in the KAN estimate and its upper bound is used within our KAN-rCBF constraints. The error shows that the estimator closely tracks the true temperature during the charging process, and it achieves a root mean square error (RMSE) of $0.0417K$. The estimation error stays within the analytical bound from Proposition~\ref{main_prop1} throughout the operating range. This shows that the theoretical guarantee is reflected in practice and that the added correction term captures the uncertainty effectively. This result also shows that the estimated temperature is reliable for use in our proposed KAN-rCBF charging strategy.
\begin{figure}[h!]
    \centering
    \includegraphics[width=.65\linewidth]{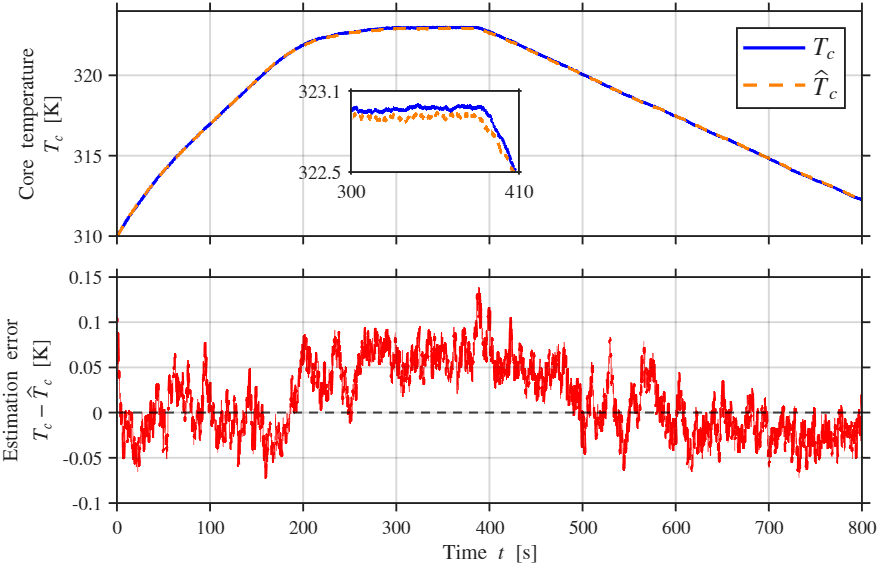}
    \caption{The upper plot compares the true core temperature $T_c$ and the estimated temperature $\widehat{T}_c$. The lower plot shows the estimation error $T_c-\widehat{T}_c$.  } \label{error} 
\end{figure}

\textbf{Performance of Baseline:} The S-CBF methodology enforces safety constraints on terminal voltage, battery temperature, and state of charge (SoC) through a cascade structure \cite{feng2024safe}. Nevertheless, there are a few shortcomings of this method. \textbf{(i)} Most important among them is the lack of a core temperature constraint, which makes it vulnerable to internal overheating, since the battery core can reach a higher temperature than the surface \cite{li2024internal}. 
\textbf{(ii)} The framework assumes perfect knowledge of the battery states and does not consider model uncertainty, which reduces its robustness for real-world applications. \textbf{(iii)}  The model used for deriving the S-CBF constraints also neglects the reversible heat source term $\mathcal{E}$. This simplification enables the use of a cascade CBF structure to separate electrical and thermal safety constraints into sequential QP problems with linear constraints. We have assessed that such a simplification is unnecessary and can also lead to poorer thermal management. Specifically, both the cascade CBF and the standard single-step CBF formulations were tested for our method, and they showed similar results in terms of charging behavior and satisfaction of the safety constraints. 
Fig.~\ref{fig:placeholder} shows the simulation results under the $10C$ fast-charging reference and includes the S-CBF optimal charging current, terminal voltage, state-of-charge (SoC), surface temperature $T_s$, and core temperature $T_c$ (from top to bottom). Since the S-CBF approach does not directly have control over the core temperature, the core temperature reaches about $327.97K$ in this case, which is nearly $5K$ above the safe temperature limit. This shows that surface temperature alone does not indicate the internal thermal condition of the battery. In addition, the S-CBF strategy fails to remain below the safe surface temperature limit, as $T_s$ reaches a maximum of $323.11K$. This occurs because the algorithm is not robust to model uncertainty. Overall, this S-CBF control approach does not consistently ensure thermal safety.

\textbf{Performance of KAN-rCBF:}
Fig.~\ref{fig:placeholder} also shows that the optimal fast charging policy obtained from the proposed KAN-rCBF strategy can satisfy the thermal constraints for the battery. Even in the presence of system uncertainty, the core temperature of the battery attains a maximum of $323K$, whereas 
the surface temperature remains within safe limits, with a maximum temperature of around $319.06K$.  We observe the major difference between S-CBF and KAN-rCBF starting around 171.60 seconds when the charging current for KAN-rCBF is reduced as the core temperature approaches its safety limit. In contrast, S-CBF continues to charge at a higher C-rate as it only monitors the surface temperature. This results in a slightly longer time of  $784.54$ seconds  to fully charge the battery using KAN-rCBF,  compared to $712.96$ seconds for the S-CBF strategy. This charging time difference of approximately 1.2 minute is a necessary trade-off to maintain safe internal temperatures.
\begin{figure}[ht!] \label{comparison}
    \centering 
    \includegraphics[width=.6\linewidth]{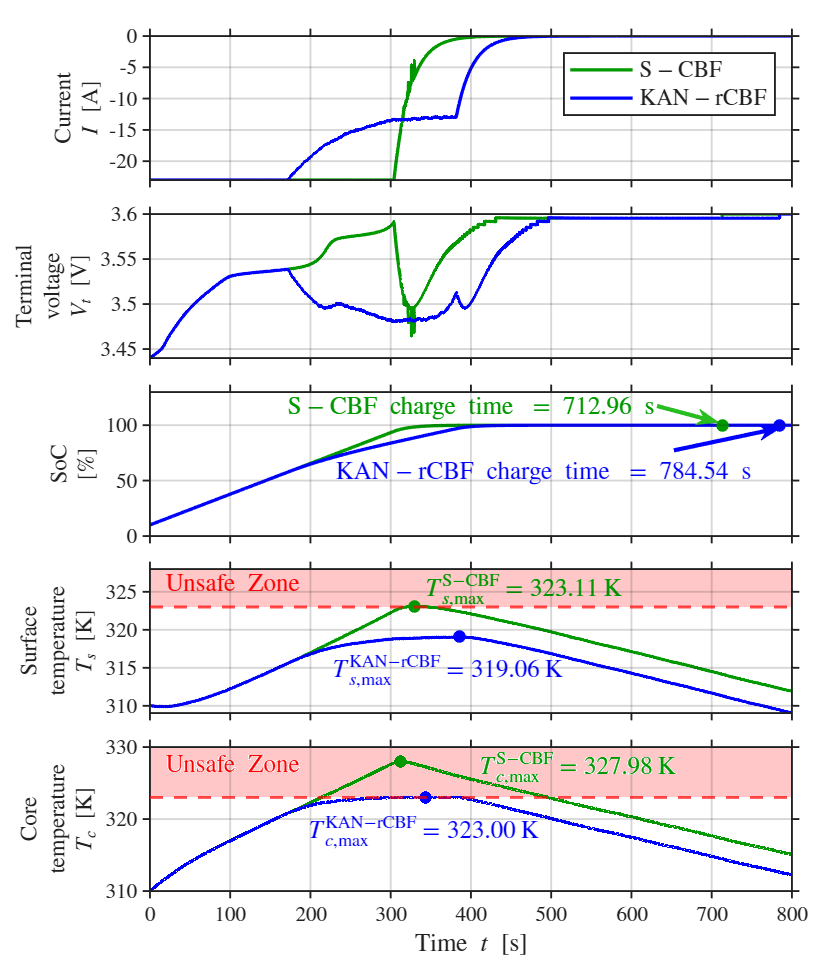}
    \caption{Performance comparison of baseline S-CBF \cite{feng2024safe} and proposed KAN-rCBF charging algorithm, where the latter achieves improved safety in temperature regulation while maintaining comparable charging time.}
    \label{fig:placeholder}
\end{figure}

\section{Conclusion And Future Work} \label{conclu}
This work develops a robust battery safety control framework that considers model uncertainty and uses KAN for core temperature estimation to explicitly include a core temperature constraint within the CBF framework. Analytical guarantees are also developed to ensure safety in the presence of KAN estimation errors and model uncertainty. The results show that it provides a reliable estimate of the core temperature for control and safety purposes.  A comparison with a state-of-the-art surface temperature constrained CBF charging method was used as a representative example. We showed that relying only on surface temperature constraints can allow the internal temperature to exceed safe levels, with the core temperature reaching about $327.97K$ and the surface temperature about $323.11K$. By constraining the core temperature with the proposed noise-aware model, both the core and surface temperatures were maintained within safe limits, with the core temperature remaining near $323K$. This improved safety required a moderate increase in charging time, approximately $1.2$ minutes. The increase in charging time reflects the prioritization of thermal safety.

\bibliography{ref1.bib}

\appendix
\section{Proofs}
\subsection{Proof of the HOCBF Condition}
\label{proof_hocbf}
\begin{proof}
Consider the system $\dot{x}=f(x)+g(x)U+\nu(t)$ as in \eqref{modelg}.
Let $B(x)$ be a barrier function. Its first derivative along the system dynamics is
\begin{align}
\dot{B}(x)=L_fB(x)+L_gB(x)U+L_{\nu}B(x).
\end{align}
If the control input does not appear in $\dot{B}$, then $L_gB(x)=0$. Differentiating once more yields
\begin{align}
B^{(2)}(x)=L_f^2B(x)+L_gL_fB(x)\,U+L_{\nu}L_fB(x)+L_{\dot{\nu}}B(x),
\end{align}
where the uncertainty term follows from $\frac{d}{dt}(L_{\nu}B)=L_{\dot{\nu}}B$, since $\nu(t)$ depends only on time. Since the input is also absent from $B^{(2)}$, so, $L_gL_fB(x)=0$.

By induction, the input is absent from all derivatives up to order $r-1$ if
\begin{align}
L_gL_f^jB(x)=0, \quad j=0,1,\dots,r-2,
\end{align}
The input $U$ appears for the first time in the $r$-th derivative through $L_gL_f^{r-1}B(x)U$, with $L_gL_f^{r-1}B(x)\neq 0$. Hence, $B(x)$ has relative degree $r$.

Next, consider the uncertainty contribution. For any $j$ and $k$,
\begin{align}
L_{\nu^{(k)}}L_f^jB(x)=\frac{\partial L_f^jB(x)}{\partial x}\nu^{(k)}(t).
\end{align}
Since $\nu(t)$ depends only on time, differentiation only increases the order of the uncertainty:
\begin{align}
\frac{d}{dt}\big(L_{\nu^{(k)}}L_f^jB(x)\big)=L_{\nu^{(k+1)}}L_f^jB(x).
\end{align}
Therefore, repeated differentiation produces the sequence
\begin{align}
L_{\nu}L_f^{r-1}B,\quad L_{\dot{\nu}}L_f^{r-2}B,\quad \dots,\quad L_{\nu^{(r-1)}}B.
\end{align}
Thus, the $r$-th derivative of $B(x)$ can be written as
\begin{align}
B^{(r)}(x,U)=L_f^rB(x)+L_gL_f^{r-1}B(x)U+\Delta_r(x,t),
\end{align}
where
\begin{align}
\Delta_r(x,t)=\sum_{k=0}^{r-1}L_{\nu^{(k)}}L_f^{r-1-k}B(x).
\end{align}

Define the barrier-state vector
\begin{align}
\eta(x)=\begin{bmatrix}B(x) & \dot{B}(x) & \cdots & B^{(r-1)}(x)\end{bmatrix}^{\top}.
\end{align}
Then
\begin{align}
\dot{\eta}=F\eta+G\mu, \quad \mu=B^{(r)}(x,U),
\end{align}
where
\begin{align}
F=
\begin{bmatrix}
0 & 1 & 0 & \cdots & 0\\
0 & 0 & 1 & \cdots & 0\\
\vdots & \vdots & \vdots & \ddots & \vdots\\
0 & 0 & 0 & \cdots & 1\\
0 & 0 & 0 & \cdots & 0
\end{bmatrix},
\quad
G=
\begin{bmatrix}
0\\
0\\
\vdots\\
0\\
1
\end{bmatrix}.
\end{align}
Choose $K_{\alpha}=\begin{bmatrix}\alpha_1 & \alpha_2 & \cdots & \alpha_r\end{bmatrix}$ such that $F-GK_{\alpha}$ is Hurwitz. The higher-order barrier condition is imposed as
\begin{align}
B^{(r)}(x,U)\geqslant -K_{\alpha}\eta(x).
\end{align}
Substituting $B^{(r)}(x,U)$ gives
\begin{align}
L_f^rB(x)+L_gL_f^{r-1}B(x)U+\Delta_r(x,t)\geqslant -K_{\alpha}\eta(x).
\end{align}

Assume now that $\Delta_r(x,t)\geqslant -\overline{\Delta}_r$. Then
\begin{align}
L_f^rB(x)+&L_gL_f^{r-1}B(x)U+\Delta_r(x,t)\geqslant \nonumber \\
&L_f^rB(x)+L_gL_f^{r-1}B(x)U-\overline{\Delta}_r.
\end{align}
Therefore, if
\begin{align}
L_f^rB(x)+L_gL_f^{r-1}B(x)U-\overline{\Delta}_r\geqslant -K_{\alpha}\eta(x), \label{rrho}
\end{align}
it follows that
\begin{align}
L_f^rB(x)+L_gL_f^{r-1}B(x)U+\Delta_r(x,t)\geqslant -K_{\alpha}\eta(x).
\end{align}
This shows that the robust HOCBF condition \eqref{rrho} is sufficient to guarantee the desired HOCBF inequality.
\end{proof}

\subsection{Proof of Proposition}\label{proof}
To establish forward invariance of the safe set $\mathcal{S}$, we will show that each of the barrier functions $B_i, \forall i\in\{1,2,3,4\}$ \eqref{V11}-\eqref{Ts} satisfy the rCBF conditions. 

\vspace{2mm}
\noindent
\textbf{Voltage constraint:} 
Since the control input $u$ appears explicitly in $B_1(x,u)$, the relative degree of $B_1(x,u)$ is $0$. Hence, using \eqref{V11} the rCBF condition  $B_1(x,u)\geqslant 0$ yields our first rCBF constraint \eqref{opt_constrain}.

\vspace{2mm}

\noindent
\textbf{SoC constraint:} 
The gradient of the safety function \eqref{soc s} with respect to the state vector $x$ is given by
$\nabla B_2(x) = [-1,  0,  0, 0]^T$,
which shows that $\dot{B}_2(x)$ depends only on the SoC dynamics. We derive the following Lie derivatives to determine the relative degree of $B_2$:
\begin{align}
    &L_f B_2(x) = \nabla B_2^T(x) f = 0,\label{lfb2}\\ 
    &L_{g_1} B_2(x)
= (-1)\left(-\tfrac{1}{C_{\text{bat}}}\right) = \tfrac{1}{C_{\text{bat}}},\label{lg1b2}\\
&L_{g_2} B_2(x)= L_{g_3} B_2(x) =0, \, L_{\nu} B(x) = - \nu_{SoC}. \label{lg2b2}
\end{align}
Non-zero $L_{g_1} B_2(x)$ in \eqref{lg1b2} shows that the relative degree is 1 and then the rCBF condition for SoC barrier function $B_2$ \eqref{soc s} is given by
\begin{align}
  L_f B_2(x)& +  L_{g_1} B_2(x) u +L_{g_2} B_2(x) u^2   +L_{g_3} B_2(x)\dot{Q}_c +L_{\nu} B_2(x) + \alpha_1(B_2(x))    \geqslant 0  \label{soc c}
\end{align}
where $\alpha_1>0$. 
Substituting \eqref{lfb2}-\eqref{lg2b2} in \eqref{soc c}, we obtain 
\begin{align}\label{soc_con2}
    \tfrac{u}{C_{\text{bat}}} + \alpha_1(\overline{SoC}-SoC)-{\nu}_{SoC}\geqslant 0
\end{align}
We note that $-{\nu}_{SoC}\geqslant-\overline{\nu}_{SoC}$. Therefore, if our second rCBF constraint  \eqref{u_cond} is satisfied, then \eqref{soc_con2} will be guaranteed to be satisfied.

\vspace{2mm}
\noindent
\textbf{Core temperature constraint:} The barrier function \eqref{llls} is used to enforce the safety requirement so that the core temperature $T_c$ does not exceed the upper bound $\overline{T}_c$.
Similarly, we can derive that its relative degree is 1 and it satisfies the robust CBF condition \eqref{robust_co} 
\begin{align} 
    L_f B_3(x) + L_{g_1} B_3(x)u &+ L_{g_2} B_3(x)u^2 +L_{g_3} B_3(x)\dot{Q}_c +L_{\nu} B_3(x)+ \alpha_2 B_3(x) \geqslant 0, \label{zeror}
\end{align}
where $\alpha_2>0$. From \eqref{llls},  $\nabla B_3(x) =
[
0, \, -1, \, 0, \, 0
]^T$.
to produce the terms in \eqref{zeror}
\begin{align}
    L_f B_3(x)=-&\tfrac{T_s - T_c}{R_c C_c}, \, L_{g_1} B_3(x) =\tfrac{\mathcal{E} T_c}{C_c}, \, L_{g_2} B_3(x) = -\tfrac{R_0}{C_c}, \nonumber\\ & L_{g_3} B_3(x) = 0,\, L_{\nu} B_3(x)=- \nu_c.
\end{align}
Evaluating \eqref{zeror} then yields
\begin{align}
-\tfrac{T_s - T_c}{R_c C_c} + \tfrac{\mathcal{E} T_c}{C_c}u &- \tfrac{R_0}{C_c}u^2-  \nu_c + \alpha_2(\overline{T}_c-T_c) \geqslant 0. \label{nnnn}
\end{align}
Now, rearranging \eqref{nnnn} and replacing the constants from Table~\ref{tab:constants}, we obtain
\begin{align}
A_1(-T_s
+T_c)+A_2T_c u
-A_3u^2- \nu_c
+\alpha_2(\overline {T}_c-T_c)
\geqslant 0.\label{llll}
\end{align} 
The above condition depends on the core temperature state $T_c$, which is not available and thus must be replaced by its KAN estimate $\widehat{T}_c$. We replace $T_c=\widehat{T}_c + e$ in all three terms of \eqref{llll}, where $e=(T_c-\widehat{T}_c)$ is the KAN estimation error. This yields
\begin{align}
-A_1 T_s+&A_1 \widehat{T}_c+ A_2 \widehat{T}_cu-A_3u^2- \nu_c + \alpha_2(\overline{T}_c-\widehat{T}_c) \geqslant \alpha_2 e - (A_1+A_2u)e \label{35555}
\end{align}
 Choosing the $\alpha_2= A_1 >0$ in \eqref{35555} results in
\begin{align}
    H_c(\widehat{T}_c,T_s,u)\! - \nu_c \geqslant - \tfrac{\mathcal{E}}{C_c}ue, \label{vvv}
\end{align}
where $H(\widehat{T}_c,T_s,u)=-A_1 T_s+ A_2 \widehat{T}_cu-A_3u^2 + A_1 \overline{T}_c$.
Next, we will again note that $-\nu_c \geqslant - \overline \nu_c$. This implies that if 
  $  H(\widehat{T}_c,T_s,u) - \overline \nu_c \geqslant -\tfrac{\mathcal{E}}{C_c}ue,$ then \eqref{vvv} will be true.

Now, using the bound from the Kolmogorov-Arnold approximation theorem \eqref{KAN-approx}, we can write $-|e|\geqslant -\mathcal{M}G^{-k-1}$  to obtain constraint \eqref{Hc_cond}.

\vspace{2mm}
\noindent
\textbf{Surface temperature constraint:} 
Similarly, the gradient of the barrier function \eqref{Ts} is $ \nabla B_4(x)=[
0, 0, -1, 0]^T.$
Using the system dynamics \eqref{modelg}, we can derive $L_{g_1}B_4(x)=L_{g_2}B_4(x)=L_{g_3}B_4(x)=0$. Additionally, $L_\nu B_4(x)=\nabla B_4(x)\nu(t) =-\nu_s(t)$, which implies that $\nabla L_\nu B_4(x)=0$ and thus its contribution will not be propagated to the next derivative of $B_4$. 
Now, we calculate 
\begin{align}
L_f B_4(x)
&=
\nabla B_4(x)f(x) =
\tfrac{T_s-T_c}{R_cC_s}
+\tfrac{T_s-T_\infty}{R_sC_s}, \label{Lfb4}
\end{align}
and its the gradient 
\begin{align}
\nabla(L_fB_4)=
\begin{bmatrix}
0 &
-\tfrac{1}{R_cC_s} &
\tfrac{1}{R_cC_s}+\tfrac{1}{R_sC_s} &
-\tfrac{1}{R_sC_s}
\end{bmatrix},
\end{align}
to obtain the terms with input for $\Ddot{B}_4$:
\begin{align} 
L_{g_1}L_{f}B_4
=\tfrac{\mathcal{E} T_c}{R_cC_cC_s}&,\quad
L_{g_2}L_{f}B_4
=-\tfrac{R_s}{R_cC_cC_s},\nonumber\\
L_{g_3}L_{f}B_4
&=\tfrac{1}{R_sC_sC_{\text{bat}}} \label{lg1}
\end{align}
Since the terms in \eqref{lg1} are non-zero, it shows that the relative for $B_4$ dynamics is 2, and the rCBF constraint \eqref{ecbf} then becomes
\begin{align}
    &L_f^2 B_4(x) +L_{g_1} L_f B_4(x)u+L_{g_2}   L_f B_4u^2 + L_{g_3}L_{f}B_4(x)\dot Q_c  \! +L_{\nu}L_f B_4(x) +\alpha_4 \left(L_f B_4(x)-\nu_s(t)\right)+ \alpha_3 B_4(x)\geqslant 0, \label{ttts}
\end{align}
where $\alpha_3, \alpha_4 > 0$ are chosen such that the polynomial $\lambda^2+\alpha_4 \lambda+\alpha_3$ is Hurwitz.
We will now calculate the remaining terms in \eqref{ttts} to obtain:
\begin{align}
L_f^2B_4(x)&=\tfrac{T_c-T_s}{R_c^2C_sC_c}
-\left(\tfrac{1}{R_cC_s}+\tfrac{1}{R_sC_s}\right) \left(
\tfrac{T_s-T_c}{R_cC_s}
+\tfrac{T_s-T_\infty}{R_sC_s}
\right)
+\tfrac{T_\infty-T_s}{R_s^2C_sC_{\mathrm{bat}}},\label{lf2b4}\\
L_\nu L_fB_4(x)&= -\tfrac{\nu_c(t)}{R_cC_s}\!+\!\left(\tfrac{1}{R_cC_s}\!+\!\tfrac{1}{R_sC_s}\right)\!\nu_s(t)\!
-\!\tfrac{\nu_{\infty}(t)}{R_sC_s}. \label{lnlfb4}
\end{align}
Substituting \eqref{lg1}, \eqref{lf2b4}, \eqref{lnlfb4} in \eqref{ttts}

\begin{align}
&\tfrac{T_c-T_s}{R_c^2C_sC_c}\!-\!\!\left(\tfrac{1}{R_cC_s}\!\!+\!\!\tfrac{1}{R_sC_s}\right)\left(
\tfrac{T_s-T_c}{R_cC_s}\!\!+\!\!\tfrac{T_s-T_\infty}{R_sC_s}
\right)\!\!+\!\!\tfrac{T_\infty-T_s}{R_s^2C_sC_{\mathrm{bat}}}
\!\!+\!\!\tfrac{\alpha T_c}{R_cC_sC_c}u-\tfrac{R_s}{R_cC_sC_c}u^2+\tfrac{1}{R_sC_sC_{\mathrm{bat}}}\dot Q_c
-\tfrac{\nu_c(t)}{R_cC_s}\nonumber\\&\!\!+\!\!\left(\tfrac{1}{R_cC_s\!\!}+\tfrac{1}{R_sC_s}\right)\nu_s(t)-\tfrac{\nu_{\infty}(t)}{R_sC_s}
-\dot \nu_s(t)+\alpha_4\left(\tfrac{T_s-T_c}{R_cC_s}+\tfrac{T_s-T_\infty}{R_sC_s}-\nu_s(t)\right)+\alpha_3(\overline T_s-T_s)\geqslant 0.
\end{align}

Since $L_\nu B_4(x)=-\nu_s(t)$, its time derivative contributes an additional term
\begin{align}
\tfrac{d}{dt}\big(L_\nu B_4(x)\big)=-\dot \nu_s(t).
\end{align}
Now, defining the nominal part 
\begin{align}
    H_s&(T_c,T_s, T_{\infty},\dot{Q}_c,u) =  A_4 T_c + A_5 T_s + A_6 T_{\infty} \nonumber\\
   &+ A_7 + A_8 T_c u + A_9 u^2 + A_{10} \dot Q_c, 
\end{align}
where the constants $A_m,\, \forall m\in\{4,\dots,10\}$ are defined in Table~\ref{coeff}.
So, the constraint is
\begin{align}
    H_s(T_c,T_s, T_{\infty},\dot{Q}_c,u)\! -\!\tfrac{\nu_c(t)}{R_cC_s}\!+&\!\left(\tfrac{1}{R_cC_s}\!+\!\tfrac{1}{R_sC_s}\right)\!\nu_s(t)\!-\tfrac{\nu_{\infty}(t)}{R_sC_s}-\dot \nu_s(t) \geqslant 0
\end{align}
We note that the uncertainty terms are bounded below by
\begin{align}
-\tfrac{\nu_c(t)}{R_cC_s} &\geqslant -\tfrac{\overline{\nu}_c}{R_cC_s}, \\
\!\left(\tfrac{1}{R_cC_s}\!+\!\tfrac{1}{R_sC_s}\!-\!\alpha_4\right)\!\nu_s(t)\!
&\geqslant\!
-\!
\!\left|\!
\tfrac{1}{R_cC_s}\!+\!\tfrac{1}{R_sC_s}\!-\!\alpha_4
\right|\overline{\nu}_s,\\
\!-\!\tfrac{\nu_{\infty}(t)}{R_sC_s} &\geqslant -\tfrac{\overline{\nu}_{\infty}}{R_sC_s}, \\
-\dot{\nu}_s(t) &\geqslant -\overline{ d\nu_s}.
\end{align}
Therefore, a worst-case robust constraint is
\begin{align}
    H_s(T_c,T_s, T_{\infty},\dot{Q}_c,u)\!-\!\tfrac{\overline{\nu}_c}{R_cC_s}
-&\left|\tfrac{1}{R_cC_s}+\tfrac{1}{R_sC_s}-\alpha_4
\right|\overline{\nu}_s-\tfrac{\overline{\nu}_{\infty}}{R_sC_s}-\overline{ d\nu_s}
\geqslant 0. \label{pps}
\end{align}
Let $\Delta_4=\tfrac{\overline{\nu}_c}{R_cC_s}+ \left|
\tfrac{1}{R_cC_s}+\tfrac{1}{R_sC_s}-\alpha_4 \right|\overline{\nu}_s+\tfrac{\overline{\nu}_{\infty}}{R_sC_s}+\overline{ d\nu_s}$ denote the worst-case disturbance margin.
So, \eqref{pps} becomes
\begin{align}
    H_s(T_c,T_s, T_{\infty},\dot{Q}_c,u)-\Delta_4\geqslant 0
\end{align}

\vspace{2mm}
\noindent
Finally, using Theorem~\ref{theory2}, we can assert that any optimal solution of $u$ \eqref{optimization} that satisfies the constraints \eqref{opt_constrain}-\eqref{SSC} for the barrier functions $B_i$ conditions, in turn, guarantees the forward invariance of the battery safety set \eqref{Safety set}. This completes our proof.

\end{document}